%% file: main.tex
\documentclass{article}
\usepackage{color}
\usepackage[T1]{fontenc} 
\usepackage{lmodern} 
\usepackage[utf8]{inputenc} 
\usepackage{cmap} 
\usepackage{xspace} 
\usepackage{fullpage,amssymb,amsmath,amsthm,xcolor}
\usepackage[colorlinks=true,citecolor=blue,linkcolor=blue]{hyperref}
\usepackage[capitalize]{cleveref}
\input{macros}
\renewcommand{\vec}[1]{\mathbf{#1}}

\title{A Parallel Repetition Theorem for the GHZ Game} \author{Justin
  Holmgren\thanks{NTT Research.  E-mail:
    \href{mailto:justin.holmgren@ntt-research.com}{\texttt{justin.holmgren@ntt-research.com}}.
    Research conducted at Princeton University, supported in part by
    the Simons Collaboration on Algorithms and Geometry and NSF grant
    No. CCF-1714779.}  \and Ran Raz\thanks{Department of Computer
    Science, Princeton University. E-mail: \href{mailto:ranr@cs.princeton.edu}{\texttt{ranr@cs.princeton.edu}}.  Research supported by the Simons
    Collaboration on Algorithms and Geometry, by a Simons Investigator
    Award and by the National Science Foundation grants
    No. CCF-1714779, CCF-2007462.}}
\begin{document}
\maketitle

\begin{abstract}
We prove that parallel repetition of the (3-player) GHZ game reduces the value of the game polynomially fast to~0. That is, the value of the GHZ game repeated in parallel $t$ times is at most $t^{-\Omega(1)}$. Previously, only a bound of 
$\approx \frac{1}{\alpha(t)}$, where $\alpha$ is the inverse Ackermann function, was known~\cite{Verbitsky96}.

The GHZ game was recently identified by Dinur, Harsha, Venkat and Yuen
as a multi-player game where all existing techniques for proving
strong bounds on the value of the parallel repetition of the game
fail. Indeed, to prove our result we use a completely new proof
technique. Dinur, Harsha, Venkat and Yuen speculated that progress on
bounding the value of the parallel repetition of the GHZ game may lead
to further progress on the general question of parallel repetition of
multi-player games. They suggested that the strong correlations
present in the GHZ question distribution represent the ``hardest
instance'' of the multi-player parallel repetition
problem~\cite{DinurHVY17}.

Another motivation for studying the parallel repetition of the GHZ
game comes from the field of quantum information. The GHZ game, first
introduced by Greenberger, Horne and Zeilinger~\cite{Greenberger1989},
is a central game in the study of quantum entanglement and has been
studied in numerous works. For example, it is used for testing quantum
entanglement and for device-independent quantum cryptography. In such
applications a game is typically repeated to reduce the probability of
error, and hence bounds on the value of the parallel repetition of the
game may be useful.
\end{abstract}

\clearpage
\newif\ifwip
\wipfalse
\newif\ifnotwip
\ifwip\notwipfalse\else\notwiptrue\fi

\tableofcontents
\clearpage

\input{intro}
\input{overview}
\input{prelims}
\input{lemmas}
\input{linearcase}
\input{partition}
\input{pseudohardness}

\input{mainthm}
\appendix
\input{prob-prelims}
\input{fourier-prelims}
\input{taubound}

\bibliographystyle{alpha}
\bibliography{refs}

\end{document}

%% file: macros.tex
\newif\ifnotes
\notesfalse
\newcommand{\authnote}[3]{\textcolor{#3}{[{\footnotesize {\bf #1:} { {#2}}}]}}
\newcommand{\jnote}[1]{\ifnotes \authnote{J}{#1}{red} \fi}

\newcommand{\textdef}[1]{\textnormal{\textsf{#1}}}

\newtheorem{theorem}{Theorem}[section]

\newtheorem{proposition}[theorem]{Proposition}
\newtheorem{claim}[theorem]{Claim}
\newtheorem{lemma}[theorem]{Lemma}

\newtheorem{corollary}[theorem]{Corollary}

\newtheorem{fact}[theorem]{Fact}

\Crefname{importedtheorem}{Imported Theorem}{Imported Theorems}
\Crefname{theorem}{Theorem}{Theorems}
\Crefname{proposition}{Proposition}{Propositions}
\Crefname{claim}{Claim}{Claims}
\Crefname{lemma}{Lemma}{Lemmas}
\Crefname{conjecture}{Conjecture}{Conjectures}
\Crefname{corollary}{Corollary}{Corollaries}
\Crefname{construction}{Construction}{Constructions}
\Crefname{example}{Example}{Examples}
\Crefname{fact}{Fact}{Facts}
\Crefname{observation}{Observation}{Observations}
\Crefname{property}{Property}{Properties}

\newtheorem{definition}[theorem]{Definition}

\Crefname{definition}{Definition}{Definitions}
\Crefname{assumption}{Assumption}{Assumptions}
\Crefname{notation}{Notation}{Notations}

\theoremstyle{remark}

\Crefname{question}{Question}{Questions}
\Crefname{remark}{Remark}{Remarks}
\Crefname{comment}{Comment}{Comments}

\def\cG{{\cal G}}

\def\cQ{{\cal Q}}
\def\cR{{\cal R}}

\def\cU{{\cal U}}
\def\cV{{\cal V}}
\def\cW{{\cal W}}
\def\cX{{\cal X}}
\def\cY{{\cal Y}}

\def\bbF{{\mathbb F}}

\def\bbR{{\mathbb R}}

\def\bbZ{{\mathbb Z}}

\newcommand{\F}{\bbF}

\newcommand{\Z}{\bbZ}

\newcommand{\xor}{\oplus}

\newcommand{\transp}{\top}

\newcommand{\E}{\mathop{\mathbb{E}}}

\newcommand{\dtv}{d_{\mathsf{TV}}}
\newcommand{\dkl}{d_{\mathsf{KL}}}

\newcommand{\Supp}{\mathrm{Supp}}

\newcommand{\symdiff}{\Delta}

\def\binset{\{0,1\}}
\newcommand{\pmone}{\{\pm 1\}}

\newcommand{\eqdef}{\stackrel{\mathsf{def}}{=}}

\newcommand{\GHZ}{\mathsf{GHZ}}

\newcommand{\embed}{\mathsf{embed}}

\newcommand{\WIN}{\textsc{Win}}

\newcommand{\Unif}{\mathrm{Unif}}

%% file: intro.tex
\section{Introduction}
 In a {$k$-player game}, players are given
  correlated ``questions'' $(q_1, \ldots, q_k)$ sampled from a
  distribution $Q$ and must produce corresponding ``answers''
  $(a_1, \ldots, a_k)$ such that
  $(q_1, \ldots, q_k, a_1, \ldots, a_k)$ satisfy a fixed predicate~$\pi$.  Crucially, the players are not allowed to communicate
  amongst themselves after receiving their questions (but they may
  agree upon a strategy beforehand).  The \textdef{value} of the game
  is the probability with which the players can win with an optimal
  strategy.
  %
  %
  Multi-player games play a central role in theoretical computer
  science due to their intimate connection with multi-prover
  interactive proofs (MIPs)~\cite{Ben-OrGKW88}, hardness of
  approximation~\cite{FeigeGLSS91}, communication complexity~\cite{ParnafesRW97, Bar-YossefJKS04}, and the EPR paradox and non-local games~\cite{EPR,CleveHTW04} 



  One basic operation on multi-player games is \textdef{parallel
    repetition}.  In the $t$-wise parallel repetition of a game,
  question tuples $(q^{(i)}_1, \ldots, q^{(i)}_k)$ are sampled
  independently for $i \in [t]$. The $j^{th}$ player is given
  $(q^{(1)}_j, \ldots, q^{(t)}_j)$, and is required to produce
  $(a^{(1)}_j, \ldots, a^{(t)}_j)$.  The players win if for every
  $i \in [t]$, $(a^{(i)}_1, \ldots, a^{(i)}_k)$ is a winning answer
  for questions $(q^{(i)}_1, \ldots, q^{(i)}_k)$.  Parallel repetition
  was first proposed in \cite{FortnowRS94} as an intuitive attempt to
  reduce the value of a game from $\epsilon$ to $\epsilon^t$, but in
  general this is not what happens~\cite{Fortnow89,Feige91,FeigeV96,Raz11}.  The
  actual effect is far more subtle and a summary of some of the known results is
  given in \cref{table:known-results}.

  \newcommand{\numreps}{t}
  \begin{table}[h]
    \centering
    \begin{tabular}{ r | c  c }
      & Two-player games & $\ge 3$-player games \\
      \hline
      Classical & $\exp(-\Omega (t ))$~\cite{Raz98} & $O \left
                                                      (\frac{1}{\alpha(\numreps)}
                                                      \right )$~\cite{Verbitsky96} \\
      Entangled & $t^{-\Omega(1)}$~\cite{Yuen16} & $O(1)$ (trivial) \\
      Non-Signaling & $\exp \left ( - \Omega \left (t \right )\right ) $~\cite{Holenstein09} & $\Omega(1)$~\cite{HolmgrenY19}\\
    \end{tabular}
    \caption{Known bounds on the worst-case (slowest) decay for various
      values of the $t$-wise parallel repetition of a non-trivial game.
      $\alpha$ denotes the inverse Ackermann function.}
    \label{table:known-results}
  \end{table}

  Much less is known about games with three or more players than about
  two-player games.  Only very weak bounds are known on how $t$-wise
  parallel repetition decreases the value of a three-player game (as a
  function of $t$).  There is a similar gap in our understanding when
  players are allowed to share entangled state; in fact, no bounds
  here are known whatsoever in the three-player case.  If players are
  more generally allowed to use any no-signaling strategy, then there
  are in fact counterexamples (lower bounds) showing that parallel
  repetition may utterly fail to reduce the (no-signaling) value of a
  three-player game.

\subsection{The GHZ Game}
The GHZ game, which we will denote by $\cG_{\GHZ}$, is a three-player
game with query distribution $Q_{\GHZ}$ that is uniform on
$\{x \in \F_2^3 : x_1 + x_2 + x_3 = 0\}$.  To win, players are
required on input $(x_1, x_2, x_3)$ to produce $(y_1, y_2, y_3)$ such
that $y_1 \xor y_2 \xor y_3 = x_1 \lor x_2 \lor x_3$.  It is easily
verified that the value of $\cG_{\GHZ}$ is $3/4$.

Dinur et al.~\cite{DinurHVY17} identified the GHZ game as a simple
example of a game for which we do not know exponential decay bounds,
writing
\begin{quote}
  \emph{``We suspect that progress on bounding the value of the parallel
    repetition of the GHZ game will lead to further progress on the general
    question.''}
\end{quote}
and
\begin{quote}
 \emph{``We believe that the strong correlations present in the GHZ question distribution represent the ``hardest instance'' of the multiplayer parallel repetition problem. Existing techniques from the two-player case (which we leverage in this paper) appear to be incapable of analyzing games with question distributions with such strong correlations.''}
\end{quote}

The GHZ game also plays an important role in quantum information
theory and in particular in entanglement testing and
device-independent quantum cryptography.  Its salient properties are
that it is an XOR game for which quantum (entangled) players can play
perfectly, but classical players can win only with probability
strictly less than $1$~\cite{MillerS13}.  No such \emph{two}-player
game is known.  Moreover, the GHZ game has the so called, self testing
property, that all quantum strategies that achieve value~1 are
essentially equivalent. This property is important for entanglement
testing and device-independent quantum cryptography.

Prior to our work, the best known parallel repetition bound for the GHZ game was
due to Verbitsky~\cite{Verbitsky96}, who observed a connection between parallel
repetition and the density Hales-Jewett theorem from Ramsey
theory~\cite{Furstenberg1991}.  Using modern quantitative versions of this
theorem~\cite{Polymath2012}, Verbitsky's result implies a bound of approximately
$\frac{1}{\alpha(t)}$, where $\alpha$ is the inverse Ackermann function.  

We
prove a bound of $t^{-\Omega(1)}$.


%% file: overview.tex
\section{Technical Overview}
To prove our parallel repetition theorem for the GHZ game we show that for an
arbitrary strategy, even if we condition on that strategy winning in several
coordinates $i_1, \ldots, i_m$, there still exists some coordinate in which that
strategy loses with significant probability.  We consider the finer-grained
event that also specifies specific queries and answers in coordinates
$i_1, \ldots, i_m$, and abstract it out as a sufficiently dense product event
$E$ over the three players' inputs.

Given an arbitrary product event $E$ that occurs with sufficiently high
probability, we show that some coordinate of $\tilde{P} \eqdef P | E$ is hard.  We do
this in three high-level steps:
\begin{enumerate}
\item We first prove this for the simpler case in which $E$ is an affine
  subspace of $\F_2^{3 \times n}$.  In fact, we show in this case that
  \emph{many} coordinates of $\tilde{P}$ are hard.
\item We then prove that when $E$ is arbitrary, $\tilde{P}$ can be written as a
  convex combination of components $\tilde{P} | \cW$, where $\cW$ is a large
  affine subspace, with most such components ``indistinguishable'' from
  $P | \cW$.  Specifically, our main requirement is that for all sufficiently
  compressing linear functions $\phi$ on $\cW$, the KL divergence of
  $\phi(\tilde{X})$ from $\phi(X)$ is small, where we sample
  $\tilde{X} \gets \tilde{P}|\cW$ and $X \gets P | \cW$.
\item With this notion of indistinguishability, we prove that if $\tilde{P}|\cW$
  is indistinguishable from $P | \cW$, then the GHZ game (or any game with a
  constant-sized answer alphabet) is roughly as hard in every coordinate with
  query distribution $\tilde{P}|\cW$ as with $P | \cW$.
\end{enumerate}
We conclude that for many coordinates $i$, there is a significant
portion of $\tilde{P}$ for which the $i^{th}$ coordinate is hard.  We
emphasize that unlike all previous parallel repetition bounds, our
proof does \emph{not} construct a local embedding of $Q_{\GHZ}$ into
$\tilde{P}$ for general $E$.


\paragraph{Local Embeddability in Affine Subspaces}
We first show that if $E$ is any affine subspace of sufficiently low codimension
$m$ in $\F_2^{3 \times n}$, then there exist many coordinates $i \in [n]$ for which
$Q_{\GHZ}$ is locally embeddable in the $i^{th}$ coordinate of the conditional
distribution $\tilde{P}$.  In fact, it will suffice for us to consider only
affine ``power'' subspaces, i.e. of the form $w + \cV^3$ for some linear
subspace $\cV$ in $\F_2^n$ and vector $w \in \F_2^{3 \times n}$. Let
$X^1, \ldots, X^n \in \F_2^3$ denote the queries in each of the $n$ repetitions.


Our observation is that when $E$ is affine there exists a subset of
coordinates $S \subseteq [n]$ with $|S| \ge 2$ such that for any $i' \in S$, $E$
depends on $(X^i)_{i \in S}$ only via the differences
$(X^{i'} - X^i)_{i' \in S \setminus \{i\}}$.  Indeed, if
$E = E_1 \times E_2 \times E_3$ and if each $E_j$ is given by an affine equation
$(X^1_j, \ldots, X^n_j) \cdot A = b_j$ for a sufficiently ``skinny'' matrix $A$,
then by the pigeonhole principle there must exist two distinct subset row-sums
of $A$ with equal values.  By considering the symmetric difference of these
subsets, and using the fact that we are working over $\F_2$, there is a set
$S \subseteq [n]$ such that the $S$-subset row-sum of $A$ is $0$.  Thus the
value of $(X^1_j, \ldots, X^n_j) \cdot A$ is unchanged if $X^i_j$ is subtracted
from $X^{i'}_j$ for every $i \in S$.

As a result, the players can all sample
$(X^{i'} - X^i)_{i' \in S \setminus \{i\}}$ and $(X^{i'})_{i' \notin S}$, which
are independent of $X^i$, using shared randomness.  On input $X^i_j$, the
$j^{th}$ player can locally compute $(X^{i'}_j)_{i' \in S}$ from $X^i_j$ and
$(X^{i'} - X^i)$.

\paragraph{Pseudo-Affine Decompositions}
At a high level, we next show that if $E$ is an arbitrary product event (with
sufficient probability mass) then $\tilde{P}$ has a ``pseudo-affine
decomposition''.  That is, there is a partition $\Pi$ of $(\F_2^n)^3$ into
affine subspaces such that if $\cW$ is a random part of $\Pi$ (as weighted by
$\tilde{P}$), then any strategy for $\tilde{P}|\cW$ can be extended to a
strategy for $P|\cW$ that is similarly successful in expectation.

To construct $\Pi$, we prove the following
sufficient conditions for $\Pi$ to be a pseudo-affine decomposition:
\begin{itemize}
\item When $\cW$ is a random part of $\Pi$ (as weighted by $\tilde{P}$), the
  distributions $\tilde{P}|\cW$ and $P|\cW$ are indistinguishable to
  all sufficiently compressing linear distinguishers.  That is, if $\cW$ is an
  affine shift of $\cV^3$, then for all subspaces $\cU \le \cV$ of sufficiently
  small co-dimension, the distributions $\tilde{P}|\cW$ and $P |\cW$ are
  statistically close modulo $\cU^3$.
\item Each part $\cW$ of $\Pi$ is in fact an affine shift of a \emph{product}
  space $\cV^3$ for some linear space $\cV$.
\end{itemize}

We construct $\Pi$ satisfying these conditions iteratively.  Starting with the
singleton partition, as long as a random part $\cW$ of $\Pi$ has some subspace
$\cU$ for which $\tilde{P}|\cW$ and $P|\cW$ are distributed differently mod
$\cU^3$, we replace each part $\cW$ of $\Pi$ by all the affine shifts of $\cU^3$
in $\cW$.  We show that this process cannot be repeated too many times when $E$
has sufficient density.


\paragraph{Pseudorandomness Preserves Hardness}
The high-level reason these conditions suffice is because for any strategy
$f = f_1 \times f_2 \times f_3$, they enable us to refine $\Pi$ to a partition
$\Pi'_f$ such that when $X$ is sampled from $\tilde{P}|\cW'$ for a random part
$\cW'$ in $\Pi'_f$, the distribution of $f(X)$ is as if $X$ were sampled
\emph{uniformly} from $\cW' \cap E$ (i.e. with $X_1$, $X_2$, and $X_3$ mutually
independent).  Moreover, when we construct $\Pi'_f$ we partition each part $\cW$
of $\Pi$ into all affine shifts of some linear space $\cU^3$ where the
codimension of $\cU^3$ in $\cW$ is not too large.  Thus the strategy $f$ on
$\tilde{P}|\cW$ effectively has the players acting as independent (randomized)
functions of their inputs modulo $\cU$.  Such strategies generalize to
$P | \cW$ by the first property of pseudo-affine decompositions stated
above.

To construct $\Pi'_f$, we ensure that $f_1$ is uncorrelated with every affine
function on $\tilde{P}|\cW'$ when $\cW'$ is a random part of $\Pi'_f$, and then
prove the desired independence by Fourier analysis.  We construct $\Pi'_f$ by
iterative refinement of $\Pi$.  Start by considering a random part $\cW$ of
$\Pi$.  Whenever $f(X_1)$ is correlated with an affine $\F_2$-valued function
$\chi$, replace $\cW$ in $\Pi$ by $\cW \cap \chi^{-1}(0)$ and
$\cW \cap \chi^{-1}(1)$, and do this in parallel for all parts of $\Pi$.  We
show that this cannot be repeated too many times, and thus we quickly arrive at
our desired $\Pi'_F$.



%% file: prelims.tex
\section{Preliminaries}
In this section we describe some preliminary definitions that are somewhat
specific to this work.  More standard preliminaries are given in
\cref{sec:probability,sec:fourier}.
\subsection{Set Theory}
\begin{definition}\label{def:partition}
  For any set $S$, a \textdef{partition of $S$} is a pairwise disjoint set of
  subsets of $S$, whose union is all of $S$.
\end{definition}
If $\Pi$ is a partition of $S$ and $x$ is an element of $S$, we write $\Pi(x)$
to denote the (unique) element of $\Pi$ that contains $x$.  

\subsection{Linear Algebra}
If $U$ is a linear subspace of $V$, we write $U \le V$ rather than
$U \subseteq V$ to emphasize that $U$ is a subspace rather than an unstructured
subset.  

We crucially rely on the Cauchy-Schwarz inequality:
\begin{definition}[Inner Product Space]
  A \textdef{real inner product space} is a vector space $V$ over $\bbR$
  together with an operation $\langle \cdot, \cdot \rangle : V \times V \to
  \bbR$ satisfying the following axioms for all $x, y, z \in V$:
  \begin{itemize}
  \item Symmetry: $\langle x, y \rangle = \langle y, x \rangle$.
  \item Linearity in the first\footnote{Because of symmetry, this implies also
      linearity in the second argument, aka bilinearity.} argument:
    $\langle ax + by, z \rangle = a \langle x, z \rangle + b \langle y, z
    \rangle$.
  \item Positive Definiteness: $\langle x, x \rangle > 0$ if $x \neq 0$.
  \end{itemize}
\end{definition}
\begin{theorem}[Cauchy-Schwarz]
  In any inner product space, it holds for all vectors $u$ and $v$ that
  $| \langle u, v \rangle |^2 \le \langle u, u \rangle \cdot \langle v, v
  \rangle$.
\end{theorem}

\subsection{Multi-Player Games}
In parallel repetition we often work with Cartesian product sets of the form
$\cX = (\cX_1 \times \cdots \times \cX_k)^n$.  For these sets, we will use
superscripts to index the outer product and subscripts to index the inner
product.  That is, we view elements $x$ of $\cX$ as tuples $(x^1, \ldots, x^n)$,
where the $i^{th}$ component of $x^j$ is $x^j_i$.  We will also write $x_i$ to
denote the vector $(x^1_i, \ldots, x^n_i)$.  If
$\{E_i \subseteq \cX_i\}_{i \in [k]}$ is a collection of subsets indexed by
subscripts, we write $E_1 \times \cdots \times E_k$ or $\prod_{i \in [k]} E_i$
to denote the set $\{x \in \cX :\, \forall i \in [k], x_i \in E_i\}$.
Similarly, if $\cY$ is a product set $(\cY_1 \times \cdots \times \cY_k)^m$, we
say $f : \cX \to \cY$ is a product function $f_1 \times \cdots \times f_k$ if
$f(x) = y$ for $y_i = f_i(x_i)$.

\begin{definition}[Multi-player Games]
  A \textdef{$k$-player game} is a tuple $(\cX, \cY, P, W)$, where
  $\cX = \cX_1 \times \cdots \times \cX_k$ and
  $\cY = \cY_1 \times \cdots \times \cY_k$ are finite sets, $P$ is a probability
  measure on $\cX$, and $W : \cX \times \cY \to \{0,1\}$ is a ``winning
  probability'' predicate.
\end{definition}

\begin{definition}[Parallel Repetition]
  Given a $k$-player game $\cG = (\cX, \cY, Q, W)$, its \textdef{$n$-fold parallel
    repetition}, denoted $\cG^n$, is defined as the $k$-player game
  $(\cX^n, \cY^n, Q^n, W^n)$, where $W^n(x,y) \eqdef \bigwedge_{j=1}^{n} W(x^{j}, y^{j})$.
\end{definition}

\begin{definition}
  The \textdef{success probability} of a function
  $f = f_1 \times \cdots f_k : \cX \to \cY$ in a $k$-player game $\cG = (\cX, \cY, Q,
  W)$ is
  \[
    v[f](\cG) \eqdef \Pr_{x \gets Q} \Big [ W \big ((x, f(x)\big ) = 1 \Big ].
  \]
\end{definition}
\begin{definition} The \textdef{value} of a $k$-player game
  $\cG = (\cX, \cY, Q, W)$, denoted $v(\cG)$, is the maximum, over all functions
  $f = f_1 \times \cdots \times f_k : \cX \to \cY$, of
  $v[f](\cG)$.
\end{definition}
\begin{fact}
  \label{fact:randomized-value}
  Randomized strategies are no better than deterministic strategies.
  \jnote{TODO be more formal}
\end{fact}

\begin{definition}[Value in $j^{th}$ coordinate]
  If $\cG = (\cX, \cY, Q, W^n)$ is a game (with a product winning predicate),
  the \textdef{value of $\cG$ in the $j^{th}$ coordinate}, denoted $v^j(\cG)$,
  is the value of the game $(\cX, \cY, Q, W')$, where $W'(x, y) = W(x^i, y^i)$.
\end{definition}

\begin{definition}[Game with Modified Query Distribution]
  If $\cG = (\cX, \cY, Q, W)$ is a game, and $P$ is a probability measure on $\cX$, we write
  $\cG | P$ to denote the game $(\cX, \cY, P, W)$.
\end{definition}


%% file: lemmas.tex
\section{Key Lemmas}
In this section, we give some Fourier-analytic conditions (see
\cref{sec:fourier} for the basics of Fourier analysis) that imply independence
of random variables under the (parallel repeated) GHZ query distribution.

It will be convenient for us to work with probability distributions in terms of
their \emph{densities} (see \cref{sec:probability} for basic probability
definitions and notation).
\begin{definition}[Probability Densities]
  If $P : \Omega \to \bbR$ is a probability distribution with support contained
  in $A$, then the \textdef{density of $P$ in $A$} is
\begin{align*}
  \varphi : A & \to \bbR \\
  x & \mapsto |A| \cdot P(x).
\end{align*}
If $A$ is unspecified, then by default it is taken to be $\Omega$.
  %
\end{definition}

\begin{lemma}
  \label{lemma:fourier-formula}
  Let $\cV$ be a (finite) vector space over $\F_2$, let $P$ be uniform on
  $\{x \in \cV^3 : x_1 + x_2 + x_3 = 0\}$, and let $U$ be uniform on $\cV^3$.

  For any subset $E = E_1 \times E_2 \times E_3$ of $\cV^3$,
  \[
    P(E) = \sum_{\chi \in \hat{\cV}} \prod_{i \in [3]} \hat{1}_{E_i}(\chi) =  U(E) \cdot \sum_{\chi \in \hat{\cV}} \prod_{i \in [3]} \hat{\varphi}_{E_i}(\chi),
  \]
  where $\varphi_{E_i}$ denotes the density in $\cV$ of the uniform distribution
  on $E_i$.
\end{lemma}
\begin{proof}
  Let $\varphi_P$ denote the density in $\cV^3$ of $P$.  That is,
  \[
    \varphi_P(x_1, x_2, x_3) =
    \begin{cases}
      |\cV| & \text{if $x_1 + x_2 + x_3 = 0$} \\
      0 & \text{otherwise.}
    \end{cases}
  \]

  Then
  \begin{align}
    \notag P(E) &= \E_{x \gets \cV^3} \left [ \varphi_P(x) \cdot
    1_{E}(x) \right ]\\
    \label{eq:fourier-formula} &= \sum_{\chi \in \widehat{\cV^3}} \hat{\varphi}_P(\chi) \cdot
                                 \hat{1}_E(\chi). && \text{(Plancherel)}
  \end{align}

  We now compute $\hat{\varphi}_P(\chi)$ and $\hat{1}_E(\chi)$.  We start by
  noting that the dual space $\widehat{\cV^3}$ is isomorphic to $\hat{\cV}^3$.
  That is, each character $\chi \in \widehat{\cV^3}$ is of the form
  $\chi(x_1, x_2, x_3) =
  \chi_1(x_1)\chi_2(x_2)\chi_3(x_3)$ for some (uniquely
  determined) $\chi_1, \chi_2, \chi_3 \in \hat{\cV}$ and conversely, each choice
  of $\chi_1, \chi_2, \chi_3 \in \hat{\cV}$ gives rise to some
  $\chi \in \widehat{\cV^3}$.

  The Fourier transform of $\varphi_P$ is given by \jnote{todo: explain why}
  \begin{equation}
    \label{eq:pdf-P}
    \hat{\varphi}_P(\chi_1, \chi_2, \chi_3) =
    \begin{cases}
      1 & \text{if $\chi_1 = \chi_2 = \chi_3$} \\
      0 & \text{otherwise.}
    \end{cases}
  \end{equation}

  Since $E$ is a product event, the Fourier transform of
  $1_{E} : \cV^3 \to \binset$ is given by
  \begin{align}
    \notag \hat{1}_E(\chi_1, \chi_2, \chi_3)
    &= \prod_{i \in [3]} \hat{1}_{E_i}(\chi_i) \\
    \label{eq:pdf-indicator} &= U(E) \cdot \prod_{i \in [3]} \hat{\varphi}_{E_i}(\chi_i).
  \end{align}
  Substituting \cref{eq:pdf-P,eq:pdf-indicator} into \cref{eq:fourier-formula}
  concludes the proof of the lemma.
\end{proof}
\begin{corollary}
  \label{cor:fourier-prob-diff}
  With $\cV$, $P$, $E$, and $U$ as in \cref{lemma:fourier-formula},
  \[
    \left | P(E) - U(E) \right | \le \sum_{\chi \in \hat{\cV} \setminus \{1\}}
    \prod_{i \in [3]} \big | \hat{1}_{E_i}(\chi) \big |,
  \]
  where $1 \in \hat{\cV}$ denotes the trivial character.
\end{corollary}
\begin{proof}
  For any probability density function $\varphi$, we have $\hat{\varphi}(1) =
  1$, so
  \begin{align*}
    \left | P(E) - U(E) \right |
    & = U(E) \cdot \left | \frac{P(E)}{U(E)} - 1
      \right | \\
    & \le U(E) \cdot \left | \sum_{\chi \in \hat{\cV} \setminus \{1\}} \prod_{i \in
      [3]} \hat{\varphi}_{E_i}(\chi) \right | \\
    & \le \sum_{\chi \in \hat{\cV} \setminus \{1\}} \prod_{i \in [3]} \left |
      \hat{1}_{E_i}(\chi) \right |.\qedhere
  \end{align*}
\end{proof}

\begin{lemma}
  \label{lemma:product-function}
  Let $\cV$ be a (finite) vector space over $\F_2$, let $P$ be uniform on
  $\{x \in \cV^3 : x_1 + x_2 + x_3 = 0\}$, let $U$ be uniform on $\cV^3$, and
  let $X = (X_1, X_2, X_3)$ denote the identity\footnote{Specifically, with the
    formalism of random variables as functions on a sample space, we mean that
    $X$ is the identity function, mapping $(x_1, x_2, x_3)$ to
    $(x_1, x_2, x_3)$.} random variable on $\cV^3$.  Let $Y_i = Y_i(X_i)$ be a
  $\cY_i$-valued random variable for each $i \in [3]$, let
  $Y = (Y_1, Y_2, Y_3)$, and let $\cY = \cY_1 \times \cY_2 \times \cY_3$.

  Let $\cW$ be a subspace of $\cV$.  If for all $\chi \in \hat{\cW}$,
  \begin{equation}
    \label{eq:main-lemma-assumption}
    \E_{(x,y) \gets P_{X,Y}} \left [ \dtv \big ( P_{\chi(X_1) | X \in x + \cW^3,
        Y_1 = y_1}, U_{\chi(X_1) | X \in x + \cW^3} \big ) \right ] \le \epsilon,
  \end{equation}
  then
  \[
    \E_{x \gets P_X} \left [ \dtv\big ( P_{Y | X \in x + \cW^3}, U_{Y | X
        \in x + \cW^3} \big ) \right ] \le \epsilon \cdot \sqrt{|\cY_2| \cdot
      |\cY_3|}.
  \]
\end{lemma}
\begin{proof}
  For $x \in \cV^3$, we will write $\bar{x}$ to denote the set $x + \cW^3$.
  Recall that $\cV / \cW$ denotes the set of all cosets
  $\{x + \cW\}_{x \in \cV}$.  For every $i \in [3]$, every
  $\bar{x}_i \in \cV / \cW$, and every $y_i \in \cY_i$, define
  $1_{i, \bar{x}_i, y_i} : \bar{x}_i \to \binset$ to be the indicator for the
  set $Y_i^{-1}(y_i) \cap \bar{x}_i$.  Define
  $\varphi_{i, \bar{x}_i, y_i}$ to be the density (in
  $\bar{x}_i$) of the uniform distribution on $Y_i^{-1}(y_i) \cap \bar{x}_i$.
  That is,
  \[
    \begin{array}{l}
      \varphi_{i, \bar{x}_i, y_i} : \bar{x}_i \to \bbR \\
      \varphi_{i, \bar{x}_i, y_i}(x'_i) =
      \begin{cases}
        \frac{|\bar{x}_i|}{|Y_i^{-1}(y_i) \cap \bar{x}_i|} & \text{if $Y_i(x'_i) = y_i$}\\
        0 & \text{otherwise.}
      \end{cases}
    \end{array}
  \]
  $\varphi_{i, \bar{x}_i, y_i}$ is easily seen to be related to
  $1_{i, \bar{x}_i, y_i}$ as
  \[
    1_{i, \bar{x}_i, y_i} = P_{Y_i|\bar{X}_i = \bar{x}_i}(y_i) \cdot \varphi_{i,
      \bar{x}_i, y_i}.
  \]


  With this notation, our assumption that \cref{eq:main-lemma-assumption} holds
  (for all $\chi \in \hat{\cW}$) is equivalent to assuming that for all
  $\chi \in \hat{\cW} \setminus \{1\}$,
  \begin{equation}
    \label{eq:each-fourier-bounded}
    \E_{(x,y) \gets P_{X,Y}} \left [ \big | \hat{\varphi}_{1,\bar{x}_1,
        y_1}(\chi) \big | \right ] \le 2 \epsilon.
  \end{equation}
  This is because for all $\chi \in \hat{\cW} \setminus \{1\}$, the distribution
  $U_{\chi(X_1) | X \in x + \cW^3}$ is uniform on $\{\pm 1\}$.

  In general for $x \in \Supp(P_X)$, we have (by \cref{cor:fourier-prob-diff})
  that for any $y \in \cY$,
  \begin{equation}
    \label{eq:gen-prob-diff}
    \left | P_{Y|X \in x + \cW^3}(y) - U_{Y|X \in x + \cW^3}(y) \right
    | \le \sum_{\chi \in \hat{\cW} \setminus \{1\}} \prod_{i \in [3]}
    \big | \hat{1}_{i, \bar{x}_i, y_i}(\chi) \big |
  \end{equation}
  because:
  \begin{itemize}
  \item the event $E = \{Y = y\}$ is a product event
    $E_1 \times E_2 \times E_3$, where each $E_i = \{Y_i = y_i\}$ depends only
    on $X_i$ or equivalently on $X_i - x_i$,
  \item the distribution $P_{X - x | \bar{X} = \bar{x}}$ is uniform on
    $\{(\vec{w}_1, \vec{w}_2, \vec{w}_3) \in \cW^3 : \vec{w}_1 + \vec{w}_2 +
    \vec{w}_3 = 0\}$, and

  \item the distribution $U_{X - x | \bar{X} = \bar{x}}$ is uniform on
    $\{(\vec{w}_1, \vec{w}_2, \vec{w}_3) \in \cW^3 \}$.
  \end{itemize}


  Thus we have
  \begin{align*}
    2 \cdot \E_{x \gets P_X} \left [ \dtv\big ( P_{Y | X \in x + \cW^3}, U_{Y | X
    \in x + \cW^3} \big ) \right ]
    & = \E_{x \gets P_{X}} \sum_{y \in \cY}  \left | P_{Y|X \in x + \cW^3}(y) -
      U_{Y|X \in x + \cW^3}(y) \right |   \\
    & \le \E_{x \gets P_X} \sum_{y \in \cY}  \sum_{\chi \neq 1}  \prod_{i \in [3]} \left
      | \hat{1}_{i, \bar{x}_i, y_i}(\chi) \right |
    \\
    & = \E_{x \gets P_X} \sum_{y \in \cY}  \sum_{\chi \neq 1}  \prod_{i \in
      \{2, 3\}} \sqrt{\left
      | \hat{1}_{1, \bar{x}_1, y_1}(\chi) \right | \cdot \hat{1}_{i, \bar{x}_i,
      y_i}(\chi)^2}.
  \end{align*}
  Now, we apply Cauchy-Schwarz on the inner product space whose elements are
  real-valued functions of $(x, y, \chi)$, and where the inner product is
  defined by
  $\langle f, g \rangle \eqdef \E_{x \gets P_X} \sum_{y \in \cY} \sum_{\chi \neq
    1} f(x, y, \chi) \cdot g(x, y, \chi)$.  This bounds the above by
  \begin{align*}
    & \sqrt{\prod_{i \in \{2,3\}} \left (\E_{x \gets P_X}  \sum_{y \in \cY} \sum_{\chi \neq
      1} \big | \hat{1}_{1, \bar{x}_1, y_1}(\chi) \big | \cdot \hat{1}_{i,
      \bar{x}_i, y_i}(\chi)^2 \right )}  \\
    = & \sqrt{\prod_{i \in \{2,3\}} \left (\sum_{\chi \neq
        1} \sum_{y \in \cY} \E_{x \gets P_X}  \Big [\big | \hat{1}_{1, \bar{x}_1, y_1}(\chi) \big | \cdot \hat{1}_{i,
        \bar{x}_i, y_i}(\chi)^2 \Big ]\right )}.
  \end{align*}
  By the independence of $(X_1, Y_1)$ and $(X_i, Y_i)$ under $P$ for $i \in \{2,3\}$, this is
  equal to
  \begin{align*}
    & \prod_{i \in \{2,3\}} \sqrt{ \sum_{\chi \neq 1}  \sum_{y \in \cY}
      \E_{x \gets P_X} \Big [ \big | \hat{1}_{1, \bar{x}_1, y_1}(\chi) \big |
      \Big ] \cdot \E_{x \gets P_X} \left
      [ \hat{1}_{i, \bar{x}_i, y_i}(\chi)^2 \right ] } \\
    = & \prod_{i \in \{2,3\}} \sqrt{ |\cY_{5 - i}| \cdot \sum_{\chi \neq 1}  \left (\sum_{y_1 \in \cY_1}
        \E_{x \gets P_X} \Big[ \big | \hat{1}_{1, \bar{x}_1, y_1}(\chi) \big |
        \Big ] \right ) \cdot \left ( \sum_{y_i \in \cY_i} \E_{x \gets P_X} \left
        [ \hat{1}_{i, \bar{x}_i, y_i}(\chi)^2 \right ] \right ) }.
  \end{align*}
  But the function $1_{1, \bar{x}_1, y_1}$ is just
  $P_{Y_1|\bar{X}_1 = \bar{x}_1}(y_1) \cdot \varphi_{1, \bar{x}_1, y_1}$, so the
  above is
  \begin{align*}
    &     \prod_{i \in \{2,3\}} \sqrt{ |\cY_{5-i}| \cdot\sum_{\chi \neq
      1} \left ( \sum_{y_1 \in \cY_1} \E_{x \gets P_X} \left[
      P_{Y_1|\bar{X}_1=\bar{x}_1}(y_1)\cdot  \big | \hat{\varphi}_{1, \bar{x}_1,
      y_1}(\chi) \big | \right ] \right ) \cdot \left ( \sum_{y_i \in \cY_i} \E_{x \gets P_X} \left [ \hat{1}_{i,
      \bar{x}_i, y_i}(\chi)^2 \right ] \right )} \\
    &   =  \prod_{i \in \{2,3\}} \sqrt{ |\cY_{i}| \cdot\sum_{\chi \neq
      1} \left ( \sum_{y_1 \in \cY_1} \E_{x \gets P_X} \left[
      P_{Y_1|\bar{X}_1=\bar{x}_1}(y_1)\cdot  \big | \hat{\varphi}_{1, \bar{x}_1,
      y_1}(\chi) \big | \right ] \right ) \cdot \left ( \sum_{y_i \in \cY_i} \E_{x \gets P_X} \left [ \hat{1}_{i,
      \bar{x}_i, y_i}(\chi)^2 \right ] \right )}
  \end{align*}
  which by the definition of expectation is
  \[
    \prod_{i \in \{2,3\}} \sqrt{|\cY_i| \cdot \sum_{\chi \neq
        1} \left ( \E_{x,y \gets P_{X,Y}} \left[
          \big | \hat{\varphi}_{1, \bar{x}_1,
            y_1}(\chi) \big | \right ] \right ) \cdot \left ( \sum_{y_i \in \cY_i} \E_{x \gets P_X} \left [ \hat{1}_{i,
            \bar{x}_i, y_i}(\chi)^2 \right ] \right )}.
  \]

  We use \cref{eq:each-fourier-bounded} to bound this by
  \begin{align*}
    & \prod_{i \in \{2,3\}} \sqrt{2 \epsilon |\cY_i| \cdot \sum_{\chi \neq
      1}\sum_{y_i \in \cY_i} \E_{x \gets P_X} \left
      [ \hat{1}_{i, \bar{x}_i, y_i}(\chi)^2 \right ]}
    \\
    & \le \prod_{i \in \{2,3\}} \sqrt{2 \epsilon |\cY_i| \cdot \sum_{y_i \in \cY_i} \E_{x \gets P_X} \left
      [ \E_{x' \gets \bar{x}_i} \left [ 1_{i, \bar{x}_i, y_i}(x')^2 \right ]
      \right ]} && \text{(Parseval's Theorem)} \\
    & = \prod_{i \in \{2,3\}} \sqrt{2 \epsilon |\cY_i| \cdot \E_{x \gets P_X} \left
      [ \E_{x' \gets \bar{x}_i} \left [ \sum_{y_i \in \cY_i} 1_{i, \bar{x}_i, y_i}(x')^2 \right ]
      \right ]}.
  \end{align*}
  But for $y_i \neq y'_i$, the supports of $1_{i, \bar{x}_i, y_i}$ and
  $1_{i, \bar{x}_i, y'_i}$ are disjoint, so this is at most $2 \epsilon
  \sqrt{|\cY_2| \cdot |\cY_3|}$.
\end{proof}


%% file: linearcase.tex
\section{Local Embeddability in Affine Subspaces}
In this section we show that the parallel repeated GHZ query distribution has
many coordinates in which the GHZ query distribution can be locally embedded,
even conditioned on any affine event of low co-dimension.  We first recall the
notion of a local embedding.
\begin{definition}
  Let $\Sigma$ be a finite set, let $k$ and $n$ be positive integers, let $Q$ be
  a probability distribution on $\Sigma^k$, and let $\tilde{P}$ be a probability distribution on
  $\Sigma^{k \times n}$.

  We say that \textdef{$Q$ is locally embeddable in the $j^{th}$ coordinate of
    $\tilde{P}$} if there exists a probability distribution $R$ on a set
  $\mathcal{R}$ and functions
  $e_1, \ldots, e_k : \Sigma \times \mathcal{R} \to \Sigma^n$ such that when
  sampling $q \gets Q$, $r \gets R$, if $\tilde{X}$ denotes the random variable
  \[
    \tilde{X} \eqdef \left (
      \begin{array}{c}
        e_1(q_{1}, r)^\transp \\
        \vdots \\
        e_k(q_{k}, r)^\transp
      \end{array}
    \right ),
  \]
  then:
  \begin{enumerate}
  \item The probability law of $\tilde{X}$ is exactly $\tilde{P}$.
  \item It holds with probability $1$ that $\tilde{X}^j = q$.
  \end{enumerate}
\end{definition}

\begin{proposition}
  \label{prop:linear-case}
  Let $n$ and $m$ be positive integers with $m < n$.  Let $Q$ denote the GHZ
  query distribution (uniform on the set
  $\cQ = \{x \in \F_2^3 : x_1 + x_2 + x_3 = 0\}$), and let $\cW$ be an affine shift of
  $\cV^3$ for a subspace $\cV \le \F_2^{1 \times n}$ of codimension $m$ with $Q^n(\cW) >
  0$.

  Then there exist at least $n - m$ distinct values of $j \in [n]$ for which $Q$
  is locally embeddable in the $j^{th}$ coordinate of
  $\tilde{P} \eqdef Q^n | \cW$.



\end{proposition}
\begin{proof}
  Suppose otherwise.  Without loss of generality, suppose that the coordinates
  that are \emph{not} locally embeddable include the first $n' \eqdef m + 1$
  coordinates (otherwise, $\cV$ can be permuted to make this so).  That is, for
  each $j \in [n']$, $Q$ is not locally embeddable in the $j^{th}$
  coordinate of $\tilde{P}$.

  Let the defining equations for $\cV$ be written as
  \[
    \cV \eqdef \left \{ x \in \F_2^{1 \times n} : x \cdot A = 0 \right \}
  \]
  for some choice of $A \in \F_2^{n \times m}$, and let
  $\vec{v} \in \F_2^{3 \times n}$ be such that $\cW = \vec{v} + \cV^3$.

  Because $2^{n'} > 2^{m}$, the pigeonhole principle implies that there exist
  two distinct sets $S_0, S_1 \subseteq [n']$ such that
  \[
    \sum_{j \in S_0} A_j = \sum_{j \in S_1} A_j,
  \]
  where recall that $A_j$ denotes the $j^{th}$ row of $A$.  Thus, there is a
  non-empty subset $S \eqdef S_0 \symdiff S_1 \subseteq [n']$ such that
  \begin{equation}
    \label{eq:pigeonhole}
    \sum_{j \in S}
    A_j = 0.
  \end{equation}

  Fix some such $S$.  We will show that for any $j \in S$, $Q$ is locally embeddable
  in the $j^{th}$ coordinate of $\tilde{P}$, which is a contradiction.  Let $X$
  denote the $\F_2^{3 \times n}$-valued random
  variable given by the identity function.

  \begin{claim}
    \label{claim:individual-distribution}
    For any $j \in S$, the distribution $\tilde{P}_{X^j}$ is identical to $Q$
    (i.e., uniformly random on $\cQ$).
  \end{claim}
  \begin{proof}
    Let $j \in S$ be given.  It suffices to show that for every $q, q' \in \cQ$,
    there is a bijection $\Phi_{q, q'} : \cQ^n \cap \cW \to \cQ^n \cap \cW$ such
    that $x \in \cQ^n \cap \cW$ satisfies $x^j = q$ if and only if
    $y \eqdef \Phi_{q, q'}(x)$ satisfies $y^j = q'$.  Such a bijection
    $\Phi_{q, q'}$ can be constructed by defining, for all $j' \in [n]$,
    \[
      \Phi_{q,q'}(x)^{j'} =
      \begin{cases}
        x^{j'} + q' - q & \text{if $j' \in S$} \\
        x^{j'} & \text{otherwise.}
      \end{cases}
    \]
    $\Phi_{q,q'}$ clearly is an injective map from $\cQ^n$ to $\cQ^n$ and
    satisfies $\Phi_{q,q'}(x)^j = x^j + q' - q$, so the only
    remaining thing to check is that it indeed maps $\cW$ into $\cW$.  This is
    true because it preserves $x \cdot A$.  Indeed, for any $i \in [3]$,
    \begin{align*}
      \Phi_{q, q'}(x)_i \cdot A & = x_i \cdot A + \sum_{j' \in S} (q'_i -
                                  q_i) \cdot A_{j'} \\
                                & = x_i \cdot A + (q'_i - q_i) \cdot \sum_{j' \in S} A_{j'} \\
                                & = x_i \cdot A  &&
                                                    \text{(by
                                                    \cref{eq:pigeonhole}).}  \qedhere
    \end{align*}
  \end{proof}

  For any $j \in S$, let $\Delta^{(j)}$ denote the random variable
  $\big (X^{j'} - X^j \big )_{j' \in S \setminus
    \{j\}}$.

  \begin{claim}\label{claim:diff-independence}
    For any $j \in S$, it holds in $\tilde{P}$ that
    $\big (\Delta^{(j)}, X^{[n] \setminus S} \big )$ and
    $X^j$ are independent.
  \end{claim}
  \begin{proof}
    Equivalently (using the definition of $\tilde{P}$), let $E$ denote the event
    that $X \in \cW$, i.e. for all $i \in [3]$,
    \[
      (X_i - \vec{v}_i) \cdot A = 0.
    \]
    We need to show that in $P$, the random variables $X^j$ and
    $\big (\Delta^{(j)}, X^{[n] \setminus S} \big )$ are conditionally
  independent given $E$.  To show this, we rely on the following fact:
    \begin{fact}
      If $Y$ and $Z$ are any independent random variables, and if $E$ is any
      event that depends only on $Z$ (and occurs with non-zero probability),
      then $Y$ and $Z$ are conditionally independent given $E$.
    \end{fact}
    It is clear that $X^j$ and $\big (\Delta^{(j)}, X^{[n] \setminus S} \big )$
    are independent in $P$.  It is also the case that $E$ depends only on
    $\big (\Delta^{(j)}, X^{ [n] \setminus S} \big )$: $E$ is defined by the
    constraint that for all $i \in [3]$,
    \begin{align*}
      0 &= (X_i - \vec{v}_i) \cdot A \\
        &= \sum_{j' \in S}  (X^{j'}_i  -
          X^{j}_i - \vec{v}^{j'}_i) \cdot A_{j'} + \sum_{j' \in [n] \setminus S}
          ( X^{j'}_i - \vec{v}^{j'}_i) \cdot A_{j'}  && \text{(by \cref{eq:pigeonhole})}\\
        &= - \vec{v}^j_i \cdot A_j + \underbrace{\sum_{j' \in S \setminus \{j\}} (X^{j'}_i  -
          X^{j}_i - \vec{v}^{j'}_i) \cdot A_{j'}}_{\text{depends only on $\Delta^{(j)}$}} + \underbrace{\sum_{j' \in [n] \setminus S}
          ( X^{j'}_i - \vec{v}^{j'}_i) \cdot A_{j'}}_{\text{depends only on $X^{[n]
          \setminus S}$}}. && \qedhere
    \end{align*}
  \end{proof}

  We now put everything togther. Fix any $j \in S$. We construct a local
  embedding of $Q$ into the $j^{th}$ coordinate of $\tilde{P}$.  For each
  $i \in [3]$, we define
  $e_i : \F_2 \times (\F_2^{3 \times n}) \to \F_2^{1 \times n}$ such that for
  each $j' \in [n]$:
  \[
    e_i(x, r)^{j'} =
    \begin{cases}
      x & \text{if $j' = j$} \\
      x + r^{j'}_{i} - r^j_{i} & \text{if $j' \in S \setminus \{j\}$} \\
      r^{j'}_{i} & \text{if $j' \notin S$.}
    \end{cases}
  \]

  Define the distribution $P^{(\embed)}$ to be the distribution on $x \in
  \F_2^{3 \times n}$ obtained
  by independently sampling $q \gets Q$
  and $r \gets \tilde{P}$, then defining
  \[
      x \eqdef \left ( \begin{array}{c} e_1(q_1, r) \\ e_2(q_2,
                          r) \\ e_3(q_3, r)
                        \end{array} \right ).
  \]
  It clearly holds with probability $1$ that $q = x^j$.

  \begin{claim}
    $P^{(\embed)} \equiv \tilde{P}$.
  \end{claim}
  \begin{proof}
   By definition, it is immediate that:
 $P^{(\embed)}_{X^{ j}} \equiv \tilde{P}_{X^{
        j}}$ and $P^{(\embed)}_{\Delta^{(j)}, X^{[n] \setminus S}} \equiv
    \tilde{P}_{\Delta^{(j)}, X^{ [n] \setminus S}}$.

    Finally, $X$ is fully determined by $X^{ j}$ and
    $(\Delta^{(j)}, X^{[n] \setminus S})$, which are independent in both
    $P^{(\embed)}$ (because $q$ and $r$ are sampled independently in the
    definition of $P^{(\embed)}$) and $\tilde{P}$ (by
    \cref{claim:diff-independence}).
  \end{proof}
  We have constructed an embedding of $Q$ into one of the first $n'$ coordinates
  of $\tilde{P}$, which is the desired contradiction.
\end{proof}


%% file: partition.tex
\section{Decomposition Into Pseudorandom Affine Components}
\label{sec:pseudorandom-partitions}
In this section we show that if $E$ is an arbitrary event with sufficient
probability mass under $P = Q_{\GHZ}^n$, then $\tilde{P} = P | E$ can be
decomposed into components with affine support that are ``similar'' to
corresponding components of $P$.  We will call such components pseudorandom.

We say that $\Pi$ is an \textdef{affine partition of $\F_2^{3 \times n}$} to
mean that:
\begin{itemize}
\item Each part $\Pi(x)$ of $\Pi$ has the form $\vec{w}(x) + \cV(x)^3$ where
  $\cV(x)$ is a subspace of $\F_2^n$, and
\item Each $\cV(x)$ has the same dimension, which we refer to as the dimension
  of $\Pi$ and denote by $\dim(\Pi)$.  The codimension of $\Pi$ is defined to be
  $n - \dim(\Pi)$.
\end{itemize}

\begin{definition}
  \label{def:pseudorandomness}
  If $\cW$ is an affine shift of a vector space $\cV^3$ (for $\cV \le \F_2^n$),
  we say that a $\cW$-valued random variable $X$ is
  \textdef{$(m, \epsilon)$-close} to $Y$ if for all linear functions
  $\phi : \F_2^n \to \F_2^m$ we have
  $\dkl(\phi^3(X) \| \phi^3(Y)) \le \epsilon$, where $\phi^3$ denotes the
  function mapping \[ \left ( \begin{array}{l} x_1 \\ x_2 \\ x_3
                \end{array}
      \right ) \mapsto
                \left ( \begin{array}{l}
                          \phi(x_1) \\ \phi(x_2) \\ \phi(x_3)
                          \end{array}
      \right ).
    \]

  We write $d_m(X \| Y)$ to denote the minimum $\epsilon$ for which $X$ is
  $(m, \epsilon)$-close to $Y$.
\end{definition}
\jnote{added this remark:} We remark that $d_m(X\|Y)$ is a non-decreasing
function of $m$.

\begin{lemma}
  \label{lemma:pseudorandom-partition}
  Let $P$ denote the distribution $Q_{\GHZ}^n$, let $X$ be the identity random
  variable, let $E$ be an event with $P(X \in E) = e^{-\Delta}$, and let
  $\tilde{P} = P \big | (X \in E)$.  For any $\delta > 0$ and any $m \in \Z^+$, there
  exists an affine partition $\Pi$ of $\F_2^{3 \times n}$, of codimension at
  most $m \cdot \frac{\Delta}{\delta}$, such that:
    \begin{equation}
      \label{eq:good-partition-KL}
      \E_{\pi \gets \tilde{P}_{\Pi(X)}} \left [
        d_m \Big (
        \tilde{P}_{X | X \in \pi}  \Big \|
        P_{X | X \in \pi}
        \Big )
      \right ] \le \delta.
    \end{equation}
\end{lemma}
\begin{proof}
  We construct the claimed partition iteratively.  Start with the trivial
  $n$-dimensional affine partition $\Pi_0 = \{\F_2^{3 \times n}\}$.  Whenever
  $\Pi_i$ is a partition $\Pi$ for which \cref{eq:good-partition-KL} does not
  hold, there exists a function
  $\phi_i : \F_2^{3 \times n} \to \F_2^{3 \times m}$ that:
  \begin{itemize}
  \item When restricted to any part $\pi$ of $\Pi_i$, $\phi_i$ is of the form
    $\phi_{i,\pi}^3$ for some linear function
    $\phi_{i, \pi} : \F_2^n \to \F_2^m$, and
  \item
    \begin{equation}
      \label{eq:kl-contradiction}
      \dkl  \Big ( \tilde{P}_{\phi_i(X) | \Pi_i(X)} \Big \| P_{\phi_i(X) |
        \Pi_i(X)} \Big ) > \delta.
    \end{equation}
  \end{itemize}

  \jnote{added the following instead of the remark you wanted to add in
    main7-comments.  Seemed easier this way to ensure that all parts have the
    same dimension.} Without loss of generality, we additionally assume that
  each $\phi_{i, \pi}$ is ``full rank'' when restricted to $\pi$.  That is, if
  $\pi$ is an affine shift of $\cV^3$, where $\cV$ has dimension $k$, then the
  restriction of $\phi_{i, \pi}$ to $\cV$ is a linear map of rank $\min(k, m)$.
  It is clear that any $\phi_{i, \pi}$ may be modified to be full rank without
  decreasing the KL divergence of \cref{eq:kl-contradiction}.

  Then by the chain rule for KL divergences,
  \begin{equation}
    \label{eq:kl-chain-rule-1}
    \dkl \Big (
    \tilde{P}_{X | \Pi_i(X), \phi_i(X)} \Big \| P_{X | \Pi_i(X),
      \phi_i(X)}
    \Big )
    <
    \dkl \Big (
    \tilde{P}_{X | \Pi_i(X)} \Big \| P_{X | \Pi_i(X)}
    \Big )
    - \delta.
  \end{equation}
  The left-hand side of \cref{eq:kl-chain-rule-1} is equivalent to
  \[
    \dkl \Big ( \tilde{P}_{X | \Pi_{i+1}(X)} \Big \| P_{X | \Pi_{i+1}(X)} \Big )
  \]
  with
  $\Pi_{i+1} = \big \{ \pi \cap \{x : \phi_i(x) = z\} \big \}_{\pi \in \Pi_i, z
    \in \F_2^{3 \times m}}$, which is an affine partition of dimension at least
  $\dim(\Pi) - m$.

  Thus with the non-negative potential function
  \[
    \Phi(\Pi) \eqdef \dkl \Big ( \tilde{P}_{X | \Pi(X)} \Big \| P_{X |
      \Pi(X)} \Big ),
  \]
  we have $\Phi(\Pi_{i+1}) < \Phi(\Pi_i) - \delta$.  But
  $\Phi(\Pi_0) = - \ln \left ( P(E) \right ) = \Delta$, so there must exist
  $i^\star \le \frac{\Delta}{\delta}$ for which
  \cref{eq:good-partition-KL} holds with $\Pi = \Pi_{i^\star}$, which has
  co-dimension at most $m \cdot \frac{\Delta}{\delta}$.
\end{proof}

%% file: pseudohardness.tex
\section{Pseudorandomness Preserves Hardness}

\begin{proposition}
  \label{prop:pseudo-hardness}
  Let $\cW \subseteq \F_2^{3 \times n}$ be an affine shift of a linear subspace $\cV^3$
  and let $P$ be a
  the uniform distribution
  on $\{w \in \cW : w_1 + w_2 + w_3 = 0\}$, which we assume to be non-empty.
  Let $X$ denote the identity random variable, let $E = E_1 \times E_2 \times
  E_3$ be an event with $P(X \in E) = e^{-\Delta}$, and define
  $\tilde{P} \eqdef P \big | (X \in E)$.  Suppose that $\tilde{P}_X$ is
  $(\lceil \frac{1}{\delta} \rceil, \delta)$-close to
  $P_X$ as in \cref{def:pseudorandomness}, for $\delta$ satisfying
  $\delta \le \min(\frac{\Delta^2}{32} \cdot e^{-4 \Delta / \epsilon},\
  \frac{\Delta^2}{32e^2},\ 2 \epsilon^2)$.

  Then 
  for each $j \in [n]$, we have
  $v^j(\cG_{\GHZ}^n | \tilde{P}) \le v^j(\cG_{\GHZ}^n | P) + 2 \epsilon$.
\end{proposition}
  \begin{proof}
    Fix $j \in [n]$ to be any coordinate, and let
    $\tilde{f} = \tilde{f}_1 \times \tilde{f}_2 \times \tilde{f}_3 : \cW \to
    \F_2^3$ be an arbitrary strategy.  Let $Y$ denote $\tilde{f}(X)$.  

    \begin{claim}
      \label{claim:second-partition}
      There exists a subspace $\cU \le \cV$ of codimension at most
      $\lceil \frac{1}{\delta} \rceil$ such that:
      \begin{itemize}
      \item The $j^{th}$ coordinate $x^j$ of $x \in \F_2^{3 \times n}$ depends only
        on $x +{\cU^3}$.
      \item For all $\chi \in \hat{\cU}$,
        \[
          \E_{(x, y) \gets P_{X, Y}} \left [\dkl \Big ( P_{\chi(X_1) | X \in x +
              \cU^3, Y_1 = y_1} \Big \| U_{\chi(X_1) | X \in x + \cU^3} \Big )
          \right ] \le \delta,
      \]
      where $U$ denotes the uniform distribution on $\cW$.
      \end{itemize}
    \end{claim}
    \begin{proof}
      Start with $\cU_1 = \{u \in \cV : u^j = 0\}$ (this ensures that any
      subspace $\cU \le \cU_1$ satisfies the first desired property).  Define a
      potential function
      \[
        Z(\cU) \eqdef \dim(\cU) - \E_{(x, y) \gets P_{X, Y}} \big [ H(X_1 | X_1
          \in x_1
          + \cU, Y_1 = y_1) \big ],
      \]
      which is clearly non-negative.  Additionally, $Z(\cU)$ (and in particular
      $Z(\cU_1)$) is at most $1$ because for
      any subspace $\cU \le \cV$ and any $x_1 \in \cV$, the entropy chain rule
      implies
      \begin{align*}
        \E_{y \gets P_{Y | X_1 \in x_1 + \cU}} \big [ H(X_1 | X_1 \in x_1 +
        \cU, Y_1 = y_1) \big ] & = H(X_1 | X_1 \in x_1 + \cU) - H(Y_1 | X_1
                                   \in x_1 + \cU) \\
                                 & \ge \dim(\cU) - 1.
      \end{align*}
      (in the first step we used the fact that $Y_1$ is a function of $X_1$.

      For $i \ge 1$, define $\chi_{i} \in \hat{\cU}_{i} \setminus \{1\}$ to maximize
      \begin{align*}
        b_i & \eqdef \E_{(x,y) \gets P_{X,Y}} \Big [ \dkl \Big ( P_{\chi_i(X_1) |
              X \in x + \cU_{i}^3, Y_1 = y_1} \Big
              \| U_{\chi_i(X_1) | X \in x + \cU_i^3} \Big ) \Big ] \\
            & = \E_{(x,y) \gets P_{X,Y}} \Big [ \dkl \Big ( P_{\chi_i(X_1) |
              X \in x + \cU_{i}^3, Y_1 = y_1} \Big
              \| \Unif_{\pmone} \Big ) \Big ] \\
            & = \E_{(x,y) \gets P_{X,Y}} \Big [ \dkl \Big ( P_{\chi_i(X_1) | X_1
              \in x_1 + \cU_{i}, Y_1 = y_1} \Big
              \| \Unif_{\pmone} \Big ) \Big ]\\
            & = 1 - \E_{(x,y) \gets P_{X,Y}} \Big [ H \big (\chi_i(X_1)|X_1 \in
              x_1 + \cU_i, Y_1 = y_1 \big ) \Big ],
      \end{align*}
      and define $\cU_{i+1} \eqdef \{u \in \cU_i : \chi_i(u) = 1\}$.  By the
      entropy chain rule, we have $Z(\cU_{i+1}) \le Z(\cU_i) - b_i$.  \jnote{I
        know I should elaborate here / write it out...}

      Since the initial
      potential is at most $1$, and all potentials are at least $0$, there must
      be some $i^\star \le \lceil \frac{1}{\delta} \rceil$ for which
      $b_{i^\star} \le \delta$.  The corresponding $\cU_{i^\star}$ is the desired subspace of
      $\cV$.
    \end{proof}

    Now let $\cU$ be as given by \cref{claim:second-partition}.  By
    \cref{lemma:product-function}, we have
    \[
    \E_{x \gets P_X} \left [ \dtv \big ( P_{Y | X \in x + \cU^3} ,
      \prod_{i \in [3]} \ P_{Y_i | X_i \in x_i + \cU} \big ) \right ] \le
    \sqrt{2 \delta}.
  \]

  By assumption of \cref{prop:pseudo-hardness} (together with Pinsker's
  inequality), $P_{X + \cU^3}$ and $\tilde{P}_{X + \cU^3}$ are
  $\sqrt{\frac{\delta}{2}}$-close in total variational distance.  We thus have
  that
  \begin{equation}
    \label{eq:first}
    \E_{x \gets \tilde{P}_X} \left [ \dtv \big ( P_{Y | X \in x + \cU^3} ,
      \prod_{i \in [3]} \ P_{Y_i | X_i \in x_i + \cU} \big ) \right ] \le
    \sqrt{8 \delta},
  \end{equation}
  by the general fact that if $P$ and $Q$ are two distributions that are
  $\epsilon$-close in total variational distance, and if $X$ is a $B$-bounded
  random variable, then $\big | \E_P[X] - \E_Q[X] \big | \le 2 B \epsilon$.

    We now obtain a probabilistic lower bound on $P(E | X + \cU^3)$.  We first
    lower bound its log-expectation:
    \begin{align*}
      \E_{x \gets \tilde{P}_X} \Big [ - \ln P \big ( E | X \in x + \cU^3 \big ) \Big
      ]
      & = \E_{x \gets \tilde{P}_X} \Big [ \dkl \big ( \tilde{P}_{X | X \in x + \cU^3} \|
        P_{X | X \in x + \cU^3} \big ) \Big ] \\
      & \le \dkl \big ( \tilde{P}_{X} \|
        P_{X} \big ) && \text{(\cref{fact:conditional-kl})}\\
      & \le \Delta.
    \end{align*}
    Markov's inequality then implies that for any $\tau$,
    \begin{equation}
      \label{eq:markov-lb}
      \Pr_{x \gets \tilde{P}_X} \big [ P ( E | X \in x +
      \cU^3) \le \tau \big ] \le \frac{\Delta}{\ln(1 / \tau)}.
    \end{equation}

    Combining \cref{eq:markov-lb} with \cref{eq:first} and \cref{fact:conditioned-dtv}, we get
    \[
      \E_{x \gets \tilde{P}_{X}} \left [ \dtv \big ( \tilde{P}_{Y | X \in x +
          \cU^3} , \prod_{i \in [3]} \ (P | X_i \in E_i)_{Y_i | X_i \in x_i +
          \cU} \big ) \right ] \le \frac{\Delta}{\ln(1/\tau)} + \frac{4 \sqrt{2
          \delta}}{\tau}.
    \]
    Since this holds for all $\tau \in [0,1]$ and because $\delta \le \frac{\Delta^2}{32e^2}$, \cref{cor:optimum-tau} implies
    that
    \begin{align}
      \label{eq:condition-stat-distance}
      \E_{x \gets \tilde{P}_{X}} \left [ \dtv \big ( \tilde{P}_{Y | X \in x + \cU^3} ,
      \prod_{i \in [3]} \ (P | X_i \in E_i)_{Y_i | X_i \in x_i + \cU} \big ) \right ] \le
      \frac{4 \Delta}{\ln \left ( \frac{\Delta}{\sqrt{32 \delta}} \right )} \le \epsilon,
    \end{align}
    where the last inequality follows from our assumption that $\delta \le
    \frac{\Delta^2}{32} \cdot e^{-4 \Delta / \epsilon}$.

    Putting everything together, we have
    \begin{align*}
      \tilde{P}_{X + \cU^3,Y}
      & = \tilde{P}_{X + \cU^3} \tilde{P}_{Y | X + \cU^3} \\
      & \approx_{\epsilon} \tilde{P}_{X + \cU^3} \cdot \prod_{i \in [3]}
        (P | X_i \in E_i)_{Y_i | X_i + \cU} \\
      & \approx_{\sqrt{\frac{\delta}{2}}} P_{X + \cU^3} \cdot \prod_{i \in [3]}
        (P | X_i \in E_i)_{Y_i | X_i + \cU},
    \end{align*}
    where $\approx$ denotes closeness in total variational distance.

    But
    $P_{X + \cU^3} \cdot \prod_{i \in [3]} (P|X_i \in E_i)_{Y_i | X_i + \cU}$ is
    just the distribution on $(x + \cU^3, y)$ obtained by sampling
    $x \gets P_X$, $y \gets F(x)$, where $F = F_1 \times F_2 \times F_3$ is the
    following randomized strategy.  On input $x_i$, $F_i$ uses local randomness
    to sample and output $y_i \gets (P|X_i \in E_i)_{Y_i | X_i \in x_i + \cU}$.
    By \cref{fact:randomized-value}, the probability that $W(x^j, y) = 1$ (which
    is well-defined because $x^j$ is a function of $x + \cU^3$) is at most
    $v^j(\cG_{\GHZ}^n | P)$.

    We thus have
    \begin{align*}
      v^j[\tilde{f}](\cG_{\GHZ}^n|\tilde{P})
      &= \tilde{P}_{X + \cU^3, Y}\big (W(X^j,
        Y)= 1 \big ) \\
      &\le v^j(\cG_{\GHZ}^n | P) + \epsilon + \sqrt{\frac{\delta}{2}} \\
      &\le
        v^j(\cG_{\GHZ}^n | P) + 2 \epsilon.
    \end{align*}

    Since this holds for arbitrary $\tilde{f}$, we have $v^j (\cG_{\GHZ}^n |
    \tilde{P}) \le v^j(\cG_{\GHZ}^n | P) + 2 \epsilon$.
  \end{proof}

%% file: mainthm.tex
\section{Proof of Main Theorem}
\begin{theorem}
  If $\cG = (\cX, \cY, Q, W)$ denotes the GHZ game, then
  $v(\cG^n) \le n^{-\Omega(1)}$.
\end{theorem}
\begin{proof}
  Recall $v(\cG) = 3/4$.

  Let $P$ denote $Q^n$; that is $P$ is uniform on
  $\big \{(X_1, X_2, X_3) \in \F_2^{3 \times n} : X_1 + X_2 + X_3 = 0\big \}$.
  Let $E = E_1 \times E_2 \times E_3$ be any product event in
  $\F_2^{3 \times n}$ with $P(E) \ge e^{-\Delta}$ (where $\Delta$ is a parameter we
  will specify later), and let $\tilde{P}$ denote $P | E$.

  Let $\delta > 0$ be a parameter we will specify later, and let
  $m = \lceil \frac{1}{\delta} \rceil$.  Recall our definition of $d_m$
  (\cref{def:pseudorandomness}).  \Cref{lemma:pseudorandom-partition} states
  that there exists an affine partition $\Pi$ of $\F_2^{3 \times n}$, of
  codimension at most $m \cdot \frac{\Delta}{\delta}$, such that:
  \[
    \E_{\pi \gets \tilde{P}_{\Pi(X)}} \left [ d_m \Big ( \tilde{P}_{X | X \in
        \pi} \Big \| P_{X | X \in \pi} \Big ) \right ] \le \delta.
  \]
  Moreover,
  \begin{align*}
    \E_{\pi \gets \tilde{P}_{\Pi(X)}} \left [
    d_{\infty} \Big ( \tilde{P}_{X | X \in
    \pi} \Big \| P_{X | X \in \pi} \Big ) \right ] & = \dkl \Big ( \tilde{P}_{X | \Pi(X)} \Big \| P_{X | \Pi(X)} \Big ) \\
                                                   & \le \dkl \big (\tilde{P}_X
                                                     \| P_X \big ) \\
                                                   & \le \Delta.
  \end{align*}
  Markov's inequality thus implies that with probability at least $1/3$ when
  sampling $\pi \gets \tilde{P}_{\Pi(X)}$, it holds that
  $d_m \Big ( \tilde{P}_{X | X \in \pi} \Big \| P_{X | X \in \pi} \Big ) \le 3
  \delta$ and
  $d_\infty \Big ( \tilde{P}_{X | X \in \pi} \Big \| P_{X | X \in \pi} \Big )
  \le 3 \Delta$.  Call such a $\pi$ pseudorandom, and let $\cR$ denote the set
  of pseudorandom $\pi$.

  By \cref{prop:pseudo-hardness}, for each pseudorandom $\pi$ we have
  \begin{equation}
    \label{eq:relative-vj}
    v^j \big ( \cG^n | (\tilde{P}|\pi) \big ) \le v^j\big (\cG^n | (P | \pi) \big
    ) + 2 \epsilon
  \end{equation}
  as long as
  \begin{equation}
    \label{eq:delta-constraints}
    3 \delta \le \min(\frac{9\Delta^2}{32} \cdot e^{-12 \Delta / \epsilon},\
    \frac{9\Delta^2}{32e^2},\ 2 \epsilon^2),
  \end{equation}
  where $\epsilon$ is a parameter we will specify later.

  By \cref{prop:linear-case}, for each $\pi \in \Pi$ (with $P(\pi) > 0$), it
  holds for all but $m \cdot \frac{\Delta}{\delta}$ values of $j \in [n]$, we
  have $v^j \big (\cG^n \big | (P | \pi) \big ) = v(\cG) = 3/4$.  By averaging,
  there exists some $j^\star \in [n]$ such that
  \[
    \E_{\pi \gets \tilde{P}_{\Pi(X)|\Pi(X) \in \cR}} \left [
      v^{j^\star} \big (\cG^n \big | (P | \pi) \big )
    \right ] \le \frac{m \Delta}{n \delta} + \left (1 - \frac{m \Delta}{n
        \delta} \right ) \cdot \frac{3}{4},
  \]
  which is at most $7/8$ if
  \begin{equation}
    \label{eq:delta-constraint2}
    \delta \ge \frac{2 m \Delta}{n}.
  \end{equation}

  Putting everything together, we have
  \begin{align*}
    v^{j^\star} \big (\cG^n | \tilde{P} \big )
    & \le \E_{\pi \gets \tilde{P}_{\Pi(X)}} \Big [
      v^{j^\star} \big ( \cG^n | (\tilde{P}|\pi) \big ) \Big ] \\
    & \le \Pr_{\pi \gets \tilde{P}_{\Pi(X)}} \left [ \pi \notin \cR \right ] +
      \Pr_{\pi \gets \tilde{P}_{\Pi(X)}} \left [ \pi \in \cR \right ] \cdot \E_{\pi \gets \tilde{P}_{\Pi(X)|\Pi(X) \in \cR}} \Big [
      v^{j^\star} \big ( \cG^n | (\tilde{P}|\pi) \big ) \Big ] \\
    & \le \frac{2}{3} + \frac{1}{3} \cdot (\frac{7}{8} + 2 \epsilon) \\
    & \le \frac{47}{48}
  \end{align*}
  if \cref{eq:delta-constraints,eq:delta-constraint2} are satisfied and if
  $\epsilon \le \frac{1}{32}$.  Setting $\epsilon = \frac{1}{32}$,
  $\Delta = 0.0005 \ln n$, $\delta = n^{-0.4}$, $m = n^{0.4}$ ensures that these
  constraints are all satisfied for sufficiently large $n$.

  Applying \cref{lemma:parrep-criterion} below with $\rho(n) = n^{-0.0005}$ and
  $\epsilon = \frac{1}{48}$ completes the proof.
\end{proof}
\begin{lemma}[Parallel Repetition Criterion]
    \label{lemma:parrep-criterion}
    Let $\cG = (\cX, \cY, Q, W)$ be a game, and let $P$ denote $Q^n$.  Suppose
    $\rho : \Z^+ \to \bbR$ is a function with $\rho(n) \ge e^{-O(n)}$ and
    $\epsilon > 0$ is a constant such that for all
    $E = E_1 \times \cdots E_k \subseteq \cX^n$ with $P^n(E) \ge \rho(n)$ there
    exists $j$ such that $v^j\big (\cG^n | (P | E) \big ) \le 1 - \epsilon$.  Then
    \[
      v(\cG^n) \le \rho(n)^{\Omega(1)}.
    \]
  \end{lemma}
  \begin{proof}
    Fix any $f = f_1 \times \cdots \times f_k : \cX^n \to \cY^n$.  Consider the
    probability space defined by sampling ${X} \gets P^n$, and let
    ${Y} = f({X})$.  We define additional random variables
    $J_1, \ldots, J_n \in [n]$ and $Z_1, \ldots, Z_n \in \cX \times \cY$ where
    $J_1$ is an arbitrary fixed value, $Z_i \eqdef ({X}^{J_i}, {Y}^{J_i})$ for
    all $i$, and $J_{i + 1}$ depends deterministically on
    ${Z}_{\le i} \eqdef (Z_1, \ldots, Z_{i})$ as follows.  When
    ${Z}_{\le i} = {z}_{\le i}$, $J_{i + 1}$ is defined to be a value
    $j \in [n]$ that minimizes
    $P^n \big ( W(X^j, Y^j) = 1 \big | {Z}_{\le i} = {z}_{\le i} \big )$.  With
    these definitions, each event $\{Z_{\le i} = z_{\le i}\}$ is a product
    event. In particular, if $P^n({Z}_{\le i} = {z}_{\le i}) \ge \rho(n)$ then
    $P^n \big ( W(X^{J_{i+1}}, Y^{J_{i+1}}) = 1 \big | {Z}_{\le i} = {z}_{\le i}
    \big ) \le 1 - \epsilon$.

    Let $\WIN_i$ denote the event that $W(Z_i) = 1$, let $\WIN_{\le i}$ denote the
    event $\WIN_1 \land \cdots \land \WIN_i$, and let $w_i$ denote
    $P^n \big (\WIN_{\le i} \big )$.  Since $\WIN_{\le i}$ is the union of some subset of the
    $|\cX|^i \cdot |\cY|^i$ disjoint product events
    $\{{Z}_{\le i} = {z}_{\le i}\}$, we have
    \[
      \Pr_{{z}_{\le i} \gets P^n_{{Z}_{\le i} | \WIN_{\le i}}} \left [
        P^n({Z}_{\le i} = {z}_{\le i}) \ge \rho(n) \right ] \ge 1 - |\cX|^i \cdot |\cY|^i
      \cdot \frac{\rho(n)}{w_i}.
    \]
    Moreover, for all ${z}_{\le i}$ for which
    $P^n({Z}_{\le i} = {z}_{\le i}) \ge \rho(n)$, we know that
    $P^n \big ( \WIN_{i + 1} \big | {Z}_{\le i} = {z}_{\le i} \big ) \le 1 -
    \epsilon$.  Thus as long as
    $w_i \ge 2 \cdot |\cX|^i \cdot |\cY|^i \cdot \rho(n)$, we have
    \begin{align*}
      w_{i+1} & = w_i \cdot P^n ( \WIN_{i+1} | \WIN_{\le i} ) \\
              & = w_i \cdot \E_{{z}_{\le i} \gets P^n_{{Z}_{\le i} |
                \WIN_{\le i}}} \big [ P^n(\WIN_{i+1} | {Z}_{\le i} =
                {z}_{\le i} )\big ] \\
              & \le w_i \cdot \left (
                \Pr_{{z}_{\le i} \gets P^n_{{Z}_{\le i} |
                \WIN_{\le i}}} \big [ P^n \big ({Z}_{\le i} =
                {z}_{\le i} \big ) < \rho \big ] + \Pr_{{z}_{\le i} \gets P^n_{{Z}_{\le i} |
                \WIN_{\le i}}} \big [ P^n \big ({Z}_{\le i} =
                {z}_{\le i} \big ) \ge \rho \big ] \cdot (1 - \epsilon)
                \right ) \\
              & \le w_i \cdot \left ( \frac{1}{2} + \frac{1}{2} \cdot (1 -
                \epsilon) \right ) \\
              & = w_i \cdot \left ( 1 - \frac{\epsilon}{2} \right )
    \end{align*}
    Iterating this inequality as long as the condition
    $w_i \ge 2 \cdot |\cX|^i \cdot |\cY|^i \cdot \rho(n)$ is satisfied, we find
    $w_{i^\star}$ such that
    $w_{i^\star} \le \min \big ( 2 \cdot |\cX|^{i^\star} \cdot |\cY|^{i^\star}
    \cdot \rho(n), (1 - \frac{\epsilon}{2})^{i^\star} \big )$.  This is
    minimized for $i^\star = \Theta(\log \frac{1}{\rho(n)})$ or $i^\star = n$
    and gives $v(\cG^n) \le w_{i^\star} \le \rho(n)^{\Omega(1)}$.
  \end{proof}

%% file: prob-prelims.tex
\section{Probability Theory}
\label{sec:probability}
We recall the notions of probability theory that we will need.
\begin{definition}
  A \textdef{probability distribution} on a finite set $\Omega$ is a function
  $P : \Omega \to \bbR$ satisfying $P(\omega) \ge 0$ for all $\omega \in \Omega$
  and $\sum_{\omega \in \Omega} P(\omega) = 1$.  We extend the domain of $P$ to
  $2^\Omega$ by writing $P(E)$ to denote $\sum_{\omega \in E} P(\omega)$ for any
  ``event'' $E \subseteq \Omega$.
\end{definition}
\begin{definition}
  The \textdef{support} of $P : \Omega \to \bbR$ is the set
  $\{\omega \in \Omega : P(\omega) > 0\}$.
\end{definition}
\begin{definition}
  A \textdef{$\Sigma$-valued random variable} on a sample space $\Omega$ is a
  function $X : \Omega \to \Sigma$.
\end{definition}
\begin{definition}[Expectations]
  If $P : \Omega \to \bbR$ is a probability distribution and
  $X : \Omega \to \bbR$ is a random variable, the \textdef{expectation of $X$
    under $P$}, denoted $\E_P[X]$, is defined to be
  $\sum_{\omega \in \Omega} P(\omega)\cdot X(\omega)$.
\end{definition}

We refer to subsets of $\Omega$ as \textdef{events}.  We use standard shorthand
for denoting events.  For instance, if $X$ is a $\Sigma$-valued random variable
and $x \in \Sigma$, we write $X = x$ to denote the event
$\{\omega \in \Omega : X(\omega) = x\}$.
\begin{definition}[Indicator Random Variables]
For any event $E$, we write $1_E$ to
denote a random variable defined as
\[
  1_E(\omega) =
  \begin{cases}
    1 & \text{if $\omega \in E$} \\
    0 & \text{otherwise.}
  \end{cases}
\]
\end{definition}
\begin{definition}[Independence]
  Events $E_1, \ldots, E_k \subseteq \Omega$ are said to be
  \textdef{independent} under a probability distribution $P$ if
  $P(E_1 \cap \cdots \cap E_k) = \prod_{i \in [k]} P(E_i)$.  Random variables
  $X_1, \ldots, X_k$ are said to be independent if the events
  $X_1 = x_1, \ldots, X_k = x_k$ are independent for any choice of
  $x_1, \ldots, x_k$.
\end{definition}

\begin{definition}[Conditional Probabilities]
  If $P : \Omega \to \bbR$ is a probability distribution and
  $E \subseteq \Omega$ is an event with $P(E) > 0$, then the
  \textdef{conditional distribution of $P$ given $E$} is denoted
  $(P | E) : \Omega \to \bbR$ and is defined to be
\[
  (P|E)(\omega) = \begin{cases}
    P(\omega) / P(E) & \text{if $\omega \in E$} \\
    0 & \text{otherwise.}
  \end{cases}
\]
\end{definition}

If $X$ is a random variable and $P$ is a probability distribution, we write
$P_X$ to denote the induced distribution of $X$ under $P$.  That is, $P_X(x) =
P(X = x)$.

If $E$ is an event, we write $P_{X|E}$ as shorthand for $(P | E)_X$.
%


\begin{definition}[Entropy]\label{def:entropy}
  If $P : \Omega \to \bbR$ is a probability distribution, the entropy (in nats)
  of $P$ is
  \[
    H(P) \eqdef -\sum_{\omega \in \Omega} P(\omega) \cdot \ln \big ( P(\omega) \big ).
  \]
  When $X$ is a random variable associated with a probability distribution $P$,
  we sometimes write $H(X)$ as shorthand for $H(P_X)$.
\end{definition}
\begin{definition}[Conditional Entropy]\label{def:cond-entropy}
  If $P$ is a probability measure with random variables $X$ and $Y$, we write
  \[
    H(P_{X | Y}) \eqdef \E_{y \gets P_Y} \left [ H(P_{X | Y = y}) \right ].
  \]
\end{definition}



\begin{fact}[Chain Rule of Conditional Entropy]
  For any probability measure $P$ and any random variables $X$, $Y$, it holds that
  \[
    H(P_{X|Y}) = H(P_{X,Y}) - H(P_{Y}).
  \]
\end{fact}


\subsection{Divergences}
\begin{definition}[Total Variational Distance]\label{def:dtv}
  If $P, Q : \Omega \to \bbR$ are two probability distributions, then the
  \textdef{total variational distance between $P$ and $Q$}, denoted
  $\dtv(P, Q)$, is
  \[
    \dtv(P, Q) \eqdef \max_{E \subseteq \Omega} \Big | P(E) - Q(E) \Big |.
  \]
  An equivalent definition is
  \[
    \dtv(P, Q) \eqdef \frac{1}{2} \sum_{\omega \in \Omega} \big | P(\omega) -
    Q(\omega) \big |
  \]
\end{definition}

\begin{definition}[Kullback-Leibler (KL) Divergence]
  If $P, Q : \Omega \to \bbR$ are probability distributions, the
  \textdef{Kullback-Leibler divergence} of $P$ from $Q$ is
  \[
    \dkl(P\|Q) \eqdef \sum_{\omega \in \Omega} P(\omega) \ln \left (
      \frac{P(\omega)}{Q(\omega)} \right ),
  \]
  where terms of the form $p \cdot \ln(p / 0)$ are treated as $0$ if $p = 0$ and
  $+\infty$ otherwise, and terms of the form $0 \cdot \ln(0 / q)$ are treated as
  $0$.
\end{definition}
The following relation between total variational distance and Kullback-Leiber
divergence, known as Pinsker's inequality, is of fundamental importance.
\begin{theorem}[Pinsker's Inequality]
  \label{fact:pinsker}
  For any probability distributions $P, Q : \Omega \to \bbR$, it holds that
  $\dtv(P, Q) \le \sqrt{\frac{1}{2} \dkl(P\|Q)}$.
\end{theorem}
\begin{definition}[Conditional KL Divergence]
  If $P, Q : \Omega \to \bbR$ are probability distributions and if $W$, $X$, $Y$,
  and $Z$ are
  random variables on $\Omega$, we write
  \[
    \dkl(P_{W | X} \| Q_{Y|Z}) \eqdef \E_{x \gets P_X} \left [ \dkl(P_{W | X =
        x} \| Q_{Y | Z = x} ) \right ],
  \]
  which is taken to be $+\infty$ if there exists $x$ with $P_X(x) > 0$ but
  $Q_Z(x) = 0$.
\end{definition}
KL divergence obeys a chain rule analogous to that for entropy.
\begin{fact}[Chain Rule for KL Divergence]
  If $P, Q : \Omega \to \bbR$ are probability distributions and $W,X,Y,Z$ are
  random variables on $\Omega$, then
  \[
    \dkl(P_{W,X}\|Q_{Y,Z}) = \dkl(P_X \| Q_Z) + \dkl(P_{W|X} \| P_{Y|Z}).
  \]
\end{fact}

\subsection{Conditional KL Divergence}
\begin{fact}
  If $P : \Omega \to \bbR$ is a probability distribution and $E \subseteq \Omega$ is an event, then
  \[
    \dkl \big (P |E \big \| P \big ) = \ln \left ( \frac{1}{P(E)} \right ).
  \]
\end{fact}
\begin{fact}
  \label{fact:conditional-kl}
  Let $P, Q : \Omega \to \bbR$ be probability distributions and let $X$, $Y$ be
  random variables on $\Omega$ with $Y$ a function of $X$.  Then
  \[
     \dkl ( P_{X|Y} \| Q_{X | Y}) \big ] \le
      \dkl(P_X \| Q_X).
  \]
\end{fact}
\begin{proof}
  This is well known, but for completeness:
  \begin{align*}
    \dkl ( P_{X
    | Y} \| Q_{X | Y})
    & = \dkl(P_{X,Y}\|Q_{X,Y}) -
      \dkl(P_Y\|Q_Y) && \text{(chain rule)} \\
    & = \dkl(P_{X} \| Q_X) - \dkl(P_Y \|Q_Y) && \text{($Y$ is a function of $X$)} \\
    & \le \dkl(P_X \| Q_X). && \text{(non-negativity of KL)} \qedhere
  \end{align*}
\end{proof}

\subsection{Conditional Statistical Distance}
\begin{fact}\label{fact:conditioned-dtv}
  Let $P, Q : \Omega \to \bbR$ be probability distributions, and let
  $E \subseteq \Omega$ be an arbitrary event.  Then
  \[
    \dtv(P|E, Q|E) \le \frac{2 \cdot \dtv(P,Q)}{P(E)}.
  \]
\end{fact}
\begin{proof}
  Suppose for the sake of contradiction that for some $A \subseteq E$, we
  have
  \[
    \left | (P|E)(A) - (Q|E)(A) \right | > \frac{2 \dtv(P,Q)}{P(E)}.
  \]
  Multiplying on both sides by $P(E)$, we obtain
  \[
    \left | P(A) - P(E) \cdot (Q | E)(A) \right | > 2 \dtv(P,Q).
  \]
  Since $|P(E) - Q(E)| \le \dtv(P,Q)$ and $(Q|E)(A) \le 1$, we have
  \[
    \left | P(A) - Q(A) \right | > \dtv(P, Q),
  \]
  which is a contradiction.
\end{proof}

\begin{corollary}
  \label{fact:expectation-quotients}
  Let $P : \Omega \to \bbR$ be a probability distribution, let $X$, $Y$ and $Z$
  be random variables on $\Omega$, and let $E$ be an event such that
  $\Pr_{z \gets P_Z} [ P(E | Z = z) \ge \delta ] \ge 1 - \tau$, and let
  $\tilde{P}$ denote $P | E$.  Then
  \[ \E_{z \gets P_Z} \big [ \dtv(\tilde{P}_{X | Z = z}, \tilde{P}_{Y | Z=z}) \big ] \le \tau +
    \frac{2 \cdot \E_{z \gets P_Z} \big [ \dtv(P_{X | Z = z}, P_{Y | Z=z}) \big ]}{\delta}.
  \]
\end{corollary}
\begin{proof}
  \begin{align*}
    & \E_{z \gets P_Z} \big [ \dtv(\tilde{P}_{X | Z = z}, \tilde{P}_{Y | Z=z})
    \big ] \\
    & = \E_{z \gets P_Z} \big [ 1_{P(E | Z = z) < \delta} \cdot
      \dtv(\tilde{P}_{X | Z = z}, \tilde{P}_{Y | Z=z}) + 1_{P(E | Z = z) \ge
      \delta} \cdot \dtv(\tilde{P}_{X | Z = z}, \tilde{P}_{Y | Z=z}) \big ]\\
    & \le \tau + \E_{z \gets P_Z} \big [ 1_{P(E | Z = z) \ge
      \delta} \cdot \dtv(\tilde{P}_{X | Z = z}, \tilde{P}_{Y | Z=z}) \big ] \\
    & \le \tau + \E_{z \gets P_Z} \left [ 1_{P(E | Z = z) \ge
      \delta} \cdot \frac{2 \cdot \dtv(P_{X | Z = z}, P_{Y | Z=z})}{P(E | Z =
      z)} \right ] \\
    & \le \tau +  \frac{2 \cdot \E_{z \gets P_Z} \left [ \dtv(P_{X | Z = z},
      P_{Y | Z=z}) \right ]}{\delta}.  \qedhere
  \end{align*}
\end{proof}

%% file: fourier-prelims.tex
\section{Fourier Analysis}
\label{sec:fourier}
For any (finite) vector space $V$ over $\F_2$, the \textdef{character group of
  $V$}, denoted $\hat{V}$, is the set of group homomorphisms mapping $V$ (viewed
as an additive group) to $\pmone$ (viewed as a multiplicative group).  Each such
homomorphism is called a \textdef{character} of $V$.

We will distinguish the spaces of functions mapping from $V \to \bbR$ and
functions mapping $\hat{V} \to \bbR$ and view them as two different inner product
spaces.  For functions mapping $V \to \bbR$, we define the inner product
\[
  \langle f, g \rangle \eqdef \E_{x \gets V} \left [ f(x) {g(x)} \right ],
\]
and for functions mapping $\hat{V} \to \bbR$, we define the inner product
\[
  \langle \hat{f}, \hat{g} \rangle \eqdef \sum_{\chi \in \hat{\cV}}
  \hat{f}(\chi) \cdot {\hat{g}(\chi)}.
\]
If there is danger of ambiguity, we use
$\hat{\langle} \cdot, \cdot \hat{\rangle}$ to denote the latter inner product,
and $\hat{\|} \cdot \hat{\|}$ to denote its corresponding norm.

\begin{fact}
  Given a choice of basis for $V$, there is a canonical isomorphism between $V$
  and $\hat{V}$.  Specifically, if $V = \F_2^n$, then the characters of $V$ are
  the functions of the form
  \[
    \chi_\gamma(v) = (-1)^{\gamma \cdot v}
  \]
  for $\gamma \in \F_2^n$.
\end{fact}

\begin{definition}
For any function $f : V \to \bbR$, its \textdef{Fourier transform} is the function
$\hat{f} : \hat{V} \to \bbR$ defined by
\[
  \hat{f}(\chi) \eqdef \langle f, \chi \rangle = \E_{x \gets V} \left [ f(x) \chi(x) \right ].
\]
\end{definition}
One can verify that the characters of $V$ are orthonormal.  Together with the
assumption that $V$ is finite, we can deduce that $f$ is equal to
$\sum_{\chi \in \hat{V}} \hat{f}(\chi) \cdot \chi$.

\begin{theorem}[Plancherel]
  \label{thm:plancherel}
  For any $f, g : V \to \bbR$,
  \[
    \langle f, g \rangle = \langle \hat{f}, \hat{g} \rangle.
  \]
\end{theorem}
An important special case of Plancherel's theorem is Parseval's theorem:
\begin{theorem}[Parseval]
  \label{thm:parseval}
  For any $f : V \to \bbR$,
  \[
    \|f \| = \| \hat{f} \|.
  \]
\end{theorem}


%% file: taubound.tex
\section{Bound on Optimization Problem}
\label{sec:taubound}
Let $W : \bbR^+ \to \bbR^+$ denote the inverse of the function
$x \mapsto x \cdot e^x$ ($W$ is known in the literature as the (principal branch
of the) Lambert W function).
We rely on the following theorem:
\begin{theorem}[{\cite[Corollary 2.4]{HoorfarH00}}]
  \label{importedthm:lambert-w}
  There exists a constant $C$ (in particular,
  $C = \ln \left ( 1 + \frac{1}{e} \right )$ works) such that for all $y \ge e$,
  \[
    W(y) \le \ln y - \ln \ln y + C.
  \]
\end{theorem}

The following corollary is more directly suited to our needs.
\begin{corollary}
  \label{cor:optimum-tau}
  For any $A, B > 0$ satisfying $A \ge eB$,
  \[
    \min_{\tau \in (0,1)} \frac{A}{\ln \left (\frac{1}{\tau} \right )} +
    \frac{B}{\tau} \le \frac{4A}{\ln(A / B)}.
  \]
\end{corollary}
\begin{proof}
  The minimum is achieved (up to a factor of two) when
  $\frac{A}{\ln \left (\frac{1}{\tau} \right )} = \frac{B}{\tau}$ because
  $\frac{A}{\ln \left ( \frac{1}{\tau} \right )}$ is monotonically increasing
  with $\tau$ while $\frac{B}{\tau}$ is monotonically decreasing.  Making the
  change of variables $z = - \ln(\tau)$, this is equivalent to
  $z e^z = \frac{A}{B}$, i.e. $z = W ( \frac{A}{B})$.  This choice of $z$ (or
  equivalently $\tau$) gives
  \begin{align*}
    \frac{A}{\ln \left (\frac{1}{\tau} \right )} +
    \frac{B}{\tau}
    & = \frac{2A}{W(A/B)} \\
    & = 2B \cdot \frac{A/B}{W(A/B)} \\
    & = 2B \cdot \exp \big (W ( A/B) \big ) && \text{(Definition of $W$)} \\
    & \le \frac{2A \cdot (1 + e^{-1})}{\ln(A / B)} && \text{(\cref{importedthm:lambert-w})}\\
    & \le \frac{4A}{\ln(A/B)}. && \qedhere
  \end{align*}
\end{proof}
